\documentclass[11pt]{article}
\usepackage[utf8]{inputenc}
\usepackage{color}
\usepackage[dvips,pdfstartview=FitH,pdfpagemode=None,colorlinks=true,citecolor=blue,linkcolor=blue]{hyperref}
\usepackage{amsmath, amsthm, amssymb, mathtools, authblk}
\usepackage{float}

\newcommand{\lref}[2][]{\hyperref[#2]{#1~\ref*{#2}}}
\renewcommand{\eqref}[2][]{\hyperref[#2]{(\ref*{#2})}}
\newcommand{\bra}[1]{\{#1\}}
\newcommand{\handout}[5]{
   \renewcommand{\thepage}{#1-\arabic{page}}
   \renewcommand{\thesection}{#1.\arabic{section}}
   \noindent
   \begin{center}
   \framebox{
      \vbox{
    \hbox to 5.78in { {\bf Communication Complexity}
     	 \hfill #2 }
       \vspace{4mm}
       \hbox to 5.78in { {\Large \hfill #5  \hfill} }
       \vspace{2mm}
       \hbox to 5.78in { {\it #3 \hfill #4} }
      }
   }
   \end{center}
   \vspace*{4mm}
}

\newtheorem{theorem}{Theorem}[section]
\newtheorem{corollary}[theorem]{Corollary}
\newtheorem{lemma}[theorem]{Lemma}

\newtheorem{definition}[theorem]{Definition}
\newtheorem{claim}[theorem]{Claim}

\newtheorem{remark}[theorem]{Remark}

\newcommand{\abs}[1]{\mathify{\left| #1 \right|}}
\DeclareMathOperator*{\E}{{\mathbb E}}

\newenvironment{proof-sketch}{\noindent{\bf Sketch of Proof}\hspace*{1em}}{\qed\bigskip}
\newenvironment{proof-idea}{\noindent{\bf Proof Idea}\hspace*{1em}}{\qed\bigskip}
\newenvironment{proof-of-lemma}[1]{\noindent{\bf Proof of Lemma #1}\hspace*{1em}}{\qed\bigskip}
\newenvironment{proof-attempt}{\noindent{\bf Proof Attempt}\hspace*{1em}}{\qed\bigskip}

\newcommand{\AND}{\mathsf{AND}}

\newcommand{\wh}{\widehat}

\newcommand{\xor}{\mathsf{XOR}}
\newcommand{\XOR}{\mathsf{XOR}}
\newcommand{\EQ}{\mathsf{EQ}}
\newcommand{\OR}{\mathsf{OR}}

\newcommand{\THR}{\mathsf{THR}}
\newcommand{\ETHR}{\mathsf{ETHR}}

\newcommand{\MAJ}{\mathsf{MAJ}}

\newcommand{\UPP}{\mathsf{UPP}}

\newcommand{\OMB}{\mathsf{OMB}}
\newcommand{\OPT}{\mathsf{OPT}}
\newcommand{\sgn}{\mathrm{sgn}}

\newcommand\restr[2]{{
  \left.\kern-\nulldelimiterspace 
  #1 
  \vphantom{\big|} 
  \right|_{#2} 
  }}

\makeatletter
\def\fnum@figure{{\bf Figure \thefigure}}
\def\fnum@table{{\bf Table \thetable}}
\long\def\@mycaption#1[#2]#3{\addcontentsline{\csname
  ext@#1\endcsname}{#1}{\protect\numberline{\csname 
  the#1\endcsname}{\ignorespaces #2}}\par
  \begingroup
    \@parboxrestore
    \small
    \@makecaption{\csname fnum@#1\endcsname}{\ignorespaces #3}\par
  \endgroup}
\def\mycaption{\refstepcounter\@captype \@dblarg{\@mycaption\@captype}}
\makeatother

\newcommand{\mathify}[1]{\ifmmode{#1}\else\mbox{$#1$}\fi}
\newcommand{\bigO}O

\newcommand{\R}{{\mathbb R}}

\renewcommand{\epsilon}{\varepsilon}
\renewcommand{\phi}{\varphi}

\title{Weights at the Bottom Matter When the Top is Heavy}

\begin{document}
\author{Arkadev Chattopadhyay\thanks{Partially supported by a Ramanujan fellowship of the DST. arkadev.c@tifr.res.in}}
\author{Nikhil S.~Mande\thanks{Supported by a DAE fellowship. nikhil.mande@tifr.res.in}}
\affil{School of Technology and Computer Science, TIFR, Mumbai}
\date{}
\setcounter{Maxaffil}{0}
\renewcommand\Affilfont{\itshape\small}
\maketitle

\begin{abstract}
Proving super-polynomial lower bounds against depth-2 threshold circuits of the form $\THR \circ \THR$ is a well-known open problem that represents a frontier of our understanding in boolean circuit complexity. By contrast, exponential lower bounds on the size of $\THR \circ \MAJ$ circuits were shown by Razborov and Sherstov \cite{RS10} even for computing functions in depth-3 $\text{AC}^0$. Yet, no separation among the two depth-2 threshold circuit classes were known. In fact, it is not clear a priori that they ought to be different. In particular, Goldmann, H{\aa}stad and Razborov \cite{GHR92} showed that the class $\MAJ \circ \MAJ$ is identical to the class $\MAJ \circ \THR$.

In this work, we provide an exponential separation between $\THR \circ \MAJ$ and $\THR \circ \THR$. We achieve this by showing a function $f$ that is computed by linear size $\THR \circ \THR$ circuits and yet has exponentially large \emph{sign rank}. This, by a well-known result, implies that $f$ requires exponentially large $\THR \circ \MAJ$ circuits to be computed. Our result suggests that the sign rank method alone is unlikely to prove strong lower bounds against $\THR \circ \THR$ circuits.

The main technical ingredient of our work is to prove a strong sign rank lower bound for an $\XOR$ function. This requires novel use of approximation theoretic tools.
\end{abstract}

\newpage

\section{Introduction}
Understanding the computational power of constant-depth, unbounded fan-in threshold circuits is one of the most fundamental open problems in theoretical computer science. Despite several years of intensive research \cite{AM05,HMPST93,HG91,GHR92,Razborov92,Bruck90,KP97,KP98,Forster01,FKLMSS01,RS10,HP10,HP15,KW16,CSS16}, we still do not have strong lower bounds against depth-3 or depth-2 threshold circuits, depending on how we define threshold gates. The most natural definition of such a gate, denoted by  $\THR_{{\mathbf{w}}}$, is just a linear halfspace induced by the real weight vector $\mathbf{w}=  (w_0,w_1,\ldots,w_n) \in \mathbb{R}^{n+1}$. In other words, on an input $x \in \{-1,1\}^n$, 
\[
\THR_{\mathbf{w}}\big(x\big) = \sgn\left(w_0 + \sum_{i=1}^n w_ix_i\right).
\]
The class of all boolean functions that can be computed by circuits of depth $d$ and polynomial size, comprising such gates, is denoted by $LT_d$. The seminal work of Minsky and Papert \cite{MP69} showed that the simple function, Parity, is not in $LT_1$. While it is not hard to verify that Parity is in $LT_2$, an outstanding problem is to exhibit an explicit function that is not in $LT_2$. This problem is now a well-identified frontier for research in circuit complexity.

A natural question is how large the individual weights in the weight vector $\mathbf{w}$ need to be if we allow just integer weights. It was well-known \cite{Muroga71} that for every threshold gate with $n$ inputs, there exists a threshold representation for it that uses only integer weights of magnitude at most $2^{O(n\log n)}$. While proving a $2^{\Omega(n)}$ lower bound is not very difficult, a matching lower bound was shown only in the nineties by H{\aa}stad \cite{Hastad94}. Understanding the power of large weights vs.~small weights in the more general context of small-depth circuits has attracted attention by several works \cite{AM05,GHR92,SB91,HP10,HP15,Razborov92,HMPST93,Hof96,GK98}. More precisely, let $\widehat{LT}_d$ denote the class of boolean functions that can be computed by polynomial size and depth 
$d$ circuits comprising only of threshold gates each of whose integer weights are polynomially bounded in $n$, the number of input bits to the circuit. Interestingly, improving upon an earlier line of work \cite{CSV84,Pip87,SB91}, Goldmann, H{\aa}stad and Razborov \cite{GHR92} showed, among other things, that $LT_{d} \subseteq \widehat{LT}_{d+1}$. It also remains open to exhibit an explicit function that is not in $\widehat{LT}_3$. This is a very important frontier, as the work of Yao \cite{Yao90} and Beigel and Tarui \cite{BT94} show that the entire class $\textsf{ACC}$ is contained in the class of functions computable by quasi-polynomial size threshold circuits of small weight and depth three. By contrast, the relatively early work of Hajnal et al.~\cite{HMPST93} established the fact that the Inner-Product modulo 2 function (denoted by \textsf{IP}), that is easily seen to be in $\widehat{LT}_3$, is not in $\widehat{LT}_2$. Summarizing, we have $\widehat{LT}_2 \subseteq LT_2 \subseteq \widehat{LT}_3$. Where precisely between $\widehat{LT}_2$ and $\widehat{LT}_3$ do current techniques for lower bounds stop working?

In search of the answer to the above question, researchers have investigated the finer structure of depth-2 threshold circuits, and this has generated many new techniques that are interesting in their own right. Recall the Majority function, denoted by $\MAJ$, that outputs 1 precisely when the majority of its $n$ input bits are set to 1. It is simple to verify that $\widehat{LT}_2 = \MAJ \circ \MAJ$. Goldmann et al.~\cite{GHR92} proved two very interesting results. First, they showed that the class $\MAJ \circ \MAJ$ and $\MAJ \circ \THR$ are identical, i.e.~weights of the bottom gates do not matter if the top gate is allowed only polynomial weight. Second, they showed that $\MAJ \circ \MAJ$ is strictly contained in the class $\THR \circ \MAJ$, i.e.~the weight at the top does matter if the bottom weights are restricted to be polynomially bounded in the input length. This revealed the following structure: 
\[
\widehat{LT}_2 = \MAJ \circ \THR \subsetneq \THR \circ \MAJ \subseteq LT_2 \subseteq \widehat{LT}_3.
\]
This raised the following two questions: how powerful is the class $\THR \circ \MAJ$ and how does one prove lower bounds on the size of such circuits?

In a breakthrough work, Forster \cite{Forster01} showed that \textsf{IP} requires size $2^{\Omega(n)}$ to be computed by $\THR \circ \MAJ$ circuits. This yielded an exponential separation between $\THR \circ \MAJ$ and $\widehat{LT}_3$. This also meant that at least one of the two containments $\THR \circ \MAJ \subseteq LT_2$ and $LT_2 \subseteq \widehat{LT}_3$ is strict. While it is quite possible that both of them are strict, until now no progress on this question was made. In particular, Amano and Maruoka \cite{AM05} and Hansen and Podolskii \cite{HP10} state that separating $\THR \circ \MAJ$ from $\THR \circ \THR = LT_2$ would be an important step for shedding more light on the structure of depth-2 boolean circuits. However, as far as we know, there was no clear target function identified for the purpose of separating the two classes.

In this work, we exhibit such a function and prove that it achieves the desired separation. To state our result formally, consider the following function that is a simple adaptation of a well-known function called \textsf{ODD-MAX-BIT}, which we denote by $\text{OMB}^0_{\ell}$: it outputs $-1$ precisely if the rightmost bit that is set to $1$ occurs at an odd index. It is simple to observe that it is a linear threshold function: 
\[
\text{OMB}^0_{\ell}\big(x\big) = -1 \iff \sum_{i=1}^{\ell} (-1)^{i+1}2^i \left(1 + x_i\right) \geq 0.5
\]
Let $f_m \circ g_n : \{-1, 1\}^{m \times n} \rightarrow \{-1, 1\}$ be the composed function on $mn$ input bits, where each of the $m$ input bits to the outer function $f$ is obtained by applying the inner function $g$ to a block of $n$ bits. Then, we show the following:

\begin{theorem}  \label{thm:omb-ckt-size}
Let $F_n$ be defined on $n=2\ell^{4/3}$ bits as $\OMB^0_{\ell} \circ \OR_{\ell^{1/3} -\log l} \circ \XOR_2$. Every $\THR \circ \MAJ$ circuit computing $F_n$ needs size $2^{\Omega\left(n^{1/4}\right)}$.
\end{theorem}

To show that the above suffices to provide us with the separation of threshold circuit classes, we first observe the following: for each $x \in \{-1,1\}^n$, let $\ETHR_{\mathbf{w}}(x) = -1 \iff w_0 + w_1x_1 +\cdots+ w_n x_n = 0$. Thus, $\ETHR$ gates are also called exact threshold gates. By first observing that every function computed by a circuit of the form $\THR \circ \OR$ can also be computed by a circuit of the form $\THR \circ \AND$ with a linear blow-up in size, it follows that $F_n$ can be computed by linear size circuits of the form $\THR \circ \AND \circ \XOR_2$. Observe that each $\AND \circ \XOR_2$ is computed by an $\ETHR$ gate. Hence, $F_n$ is in $\THR \circ \ETHR$, a class that Hansen and Podolskii \cite{HP10} showed is identical to the class $\THR \circ \THR$. Thus, Theorem~\ref{thm:omb-ckt-size} yields the following fact:

\begin{corollary}  \label{cor:main}
The function $F_n$ (exponentially) separates the class $\THR \circ \MAJ$ from $\THR \circ \THR$.
\end{corollary}

\subsection{Our Techniques and Related Work}  \label{sec:techniques}
The starting point for our lower bound is the same as for all known lower bounds (see, for example, \cite{Forster01, RS10, BT16}) on the size of $\THR \circ \MAJ$ circuits. We strive to prove a lower bound on a quantity called the \emph{sign rank} of our target function $f$. Given a partition of the input bits of $f$ into two parts $X,Y$, consider the real matrix $M_f$, given by $M_f[x,y] = f(x, y)$ for each $x \in \{-1,1\}^{|X|},y \in \{-1,1\}^{|Y|}$. Any real matrix sign represents $M_f$ if each if its entries agrees in sign with the corresponding entry of $M_f$. The sign rank of $M_f$ (also informally called sign rank of $f$, when the input partition is clear from the context) is the rank of a minimal rank matrix that sign represents it. It is not hard to see that the sign rank of a function $f$ computed by $\THR \circ \MAJ$ circuit of size $s$ is at most $O(n \cdot s)$. This sets a target of proving a strong lower bound on the sign rank of $f$ for showing that it is hard for $\THR \circ \MAJ$.

Sign rank has a matrix-rigidity flavor to it, and therefore is quite non-trivial to bound. Forster's \cite{Forster01} deep result (see Theorem~\ref{thm:Forster}) shows that the sign rank of a matrix can be lower bounded by appropriately upper-bounding its spectral norm. This is enough to lower bound the sign rank of functions like \textsf{IP} as the corresponding matrices are orthogonal and therefore have relatively small spectral norm. However for other functions $f$, the spectral norm of the sign matrix $M_f$ can be large. This is true, for example, for many functions in $\text{AC}^0$. In a beautiful work, Razborov and Sherstov \cite{RS10} showed that Forster's basic method can be adapted to prove exponentially strong lower bounds on the sign rank of such a function $f$. However, our first problem is on devising an $f$ that is in $\THR \circ \THR$ that plausibly has high sign rank. On this, we were guided by another interpretation of sign rank, due to Paturi and Simon \cite{PS86}. Paturi and Simon introduced a model of 2-party randomized communication, called the unbounded-error model. In this model, Alice and Bob have to give the right answer with probability just greater than 1/2 on every input. This is, by far, the strongest 2-party known model against which we know how to prove lower bounds. \cite{PS86} showed that the sign rank of the communication matrix of $f$ essentially characterizes its unbounded error complexity. 

Why should some function $f \in \THR \circ \THR$ have large unbounded-error complexity? The natural protocol one is tempted to use is the following: assume that the sum of the magnitude of the weights of the top $\THR$ gate is 1. Sample a sub-circuit of the top gate with a probability proportional to its weight. Then, use the best protocol for the sampled bottom $\THR$ gate. Note that for any given input $x$, with probability $1/2+\epsilon$, one samples a bottom gate that agrees with the value of $f(x)$. Here, $\epsilon$ can be as small as the smallest weight of the top gate. Thus, if we had a small cost randomized protocol for the bottom $\THR$ gate that errs with probability significantly less than $\epsilon$ we would have a small cost unbounded-error protocol for our function $f$. Fortunately for us (the lower bound prover), there does not seem to exist any such efficient randomized protocol for $\THR$, when $\epsilon = 1/2^{n^{\Omega(1)}}$. 

Taking this a step further, one could hope that the bottom gates could be any function that is hard to compute with such tiny error $\epsilon$. The simplest such canonical function is Equality (denoted by $\EQ$). Therefore, a plausible target is $\THR \circ \EQ$. This still turns out to be in $\THR \circ \THR$ as $\EQ \in \ETHR$. Moreover, $\EQ$ has a nice composed structure. It is just $\AND \circ \XOR$, which lets us re-express our target as $f = \THR \circ \AND \circ \XOR$, for some top $\THR$ that is `suitably' hard; hard so that the sign rank of $f$ becomes large! At this point, we view $f$ as an $\XOR$ function whose outer function, $g$, needs to have sufficiently good analytic properties for us to prove that $g \circ \XOR$ has high sign rank.

We are naturally drawn to the work of Razborov and Sherstov \cite{RS10} for inspiration as they bound the sign rank of a three-level composed function as well. They showed that $\AND \circ \OR \circ \AND_2$, an \emph{$\AND$ function}, has high sign rank. They exploited the fact that $\AND$ functions embed inside them \emph{pattern matrices}, which have nice convenient spectral properties as observed in \cite{She11}. These spectral properties dictate them to look for an \emph{approximately smooth orthogonalizing} distribution w.r.t which the outer function $f=\AND \circ \OR$ has zero correlation with small degree parities. This gives rise naturally to an LP that seeks to maximize the smoothness of the distribution under the constraints of low-degree orthogonality. The main technical challenge that the Razborov-Sherstov work overcomes is the analysis of the dual of this LP using and building appropriate tools of approximation theory. We take cue from this work and follow its general framework of analysing the dual of a suitable LP. However, as we are forced to work with an $\XOR$ function, there are new challenges that crop up. This is expected for if we take the same outer function of $\AND \circ \OR$, then the resulting $\XOR$ function has small sign rank. Indeed, this remains true even if one were to harden the outer function to $\MAJ \circ \OR$. This is simply because a simple efficient $\UPP$ protocol for $\MAJ \circ \EQ$ exists: pick a random $\EQ$ and then execute a protocol of cost $O(\log n)$ that solves this $\EQ$ with error less than $1/n^2$. 

The specific new technical challenge that one faces is the following: instead of low degree orthogonality, one now needs a distribution $\mu$ w.r.t which the outer function has low correlation with \emph{all parities} (see \hyperlink{LP1}{LP1}).  Just dealing with high degree parity constraints, though non-trivial, was done in the recent work of the authors \cite{CM17}. However, unlike there, here one needs the additional constraint of the distribution being (approximately) smooth enough. Analysing this combination of high degree parity constraints and the smoothness constraints, is the main new technical challenge that our work addresses. We do this by a novel combination of ideas that differs entirely from the Razborov-Sherstov analysis.

Analyzing the dual of our LP (\hyperlink{LP2}{LP2}) involves arguing against the existence of a certain kind of (possibly high degree) polynomial representation. We require several ideas to deal with it. First, the dual polynomial $P$ has unit weight. While it does not necessarily sign represent $f = \THR_{\ell} \circ \OR_m$ , it is constrained to not stray too far away from zero on the wrong side on each point of its domain $\{-1,1\}^n$.  Moreover, over a set $X$, where we want the distribution to be smooth in in the primal LP, roughly speaking, $P$'s margin in representing $f$ on average has to be good. Since the set $X$ has to be large (to get good approximate smoothness), we are essentially forced to include in $X$ all inputs that are mapped to $1^l$ by the bottom $\OR$s of $f$. In particular, we set $X$ to be precisely the set of such points. With this setting, our bound on the sign rank becomes roughly $\delta/\OPT$, where $\OPT$ is the optimal value of the LP. 

The first idea we use is an averaging argument that appeared in the work of Krause and Pudl{\'a}k \cite{KP98}. What this does is that for each possible input $y \in \{-1,1\}^{\ell}$ to $\THR_{\ell}$, it takes the average of all values of $P$ under the uniform distribution over all points $x$ such that $\OR_m\big(x\big) = y$. This achieves the following as described in Lemma~\ref{lem: thror}: the polynomial $P$ over $x$ transforms to \emph{an $\OR$ polynomial} $Q$, over $y$'s, of the same weight as $P$, plus an error term whose magnitude is exponentially small in the fan-in of the bottom $\OR$ gates of $f$. Here, an $\OR$ polynomial is a linear combination of $\OR$s of subsets of variables. Assuming, for the sake of contradiction, $\OPT$ to be large enough, we can safely ignore the error term. This gives us a passage to an $\OR$ polynomial of unit weight representing our top $\THR$ function $g$, with the same worst-case guarantee that held for $P$. Additionally, we get the guarantee that at $y=-1^l$, $Q$'s margin is better by the average margin of $P$ on the set $X$. The intuition is that when $OPT$ is large, this average margin is also large compared to $\Delta$, the worst case margin.

Now we want to argue that such a $Q$ does not exist if we select our top threshold $g$ judiciously. We select the ODD-MAX-BIT function, denoted by $\OMB$, for this purpose. We then observe that if we randomly restrict each variable to $-1$, then the expected weight of $\OR$ monomials of degree at least $d$ that do not get fixed is as small as $1/2^d$. Ignoring this high degree monomials, therefore does not decrease our margin by too much. Further, with high probability, the restriction induces an $\OMB$ of sufficiently large number of free variables. This now gives us a polynomial of $Q'$ of degree less than $d$ that has worst case margin not too small, but does somewhat better on $-1^l$.  While margin bounds against sign representing polynomials of sufficiently small degree have been obtained several times before, our setting is different. $Q'$ is not sign-representing $\OMB$. It is here that our choice of the ODD-MAX-BIT function comes in very handy. We show that a standard approximation theoretic lemma of Ehlich and Zeller \cite{EZ64}, Rivlin and Cheney \cite{RC66} can be used to argue against the existence of such a $Q'$ for $\OMB$.

\section{Preliminaries}

In this section, we provide the necessary preliminaries.
\begin{definition}[Threshold functions]
A function $f : \bra{-1, 1}^n \rightarrow \bra{-1, 1}$ is called a \emph{linear threshold function} if there exist integer weights $a_0, a_1, \dots, a_n$ such that for all inputs $x \in \bra{-1, 1}^n$, $f(x) = \sgn(a_0 + \sum_{i = 1}^na_ix_i)$.
Let $\THR$ denote the class of all such functions.
\end{definition}

\begin{definition}[Exact threshold functions]
A function $f : \bra{-1, 1}^n \rightarrow \bra{-1, 1}$ is called an \emph{exact threshold function} if there exist reals $w_1, \dots, w_n, t$ such that 
\[
f(x) = -1 \iff \sum_{i = 1}^n w_ix_i = t
\]
Let $\ETHR$ denote the class of exact threshold functions.
\end{definition}

Hansen and Podolskii \cite{HP10} showed the following.

\begin{theorem}[Hansen and Podolskii \cite{HP10}]\label{thm: HP}
If a function $f : \bra{-1, 1}^n \rightarrow \bra{-1, 1}$ can be represented by a $\THR \circ \ETHR$ circuit of size $s$, then it can be represented by a $\THR \circ \THR$ circuit of size $2s$.
\end{theorem}
For the sake of completeness and clarity, we provide the proof below.
\begin{proof}
Let $f$ be an exact threshold function with the representation $\sum_{i = 1}^n w_ix_i = t$.  There exists an $\epsilon_f > 0$ such that $\sum_{i = 1}^n w_ix_i > t \implies \sum_{i = 1}^n w_ix_i > t + \epsilon_f$.
Consider a $\THR \circ \ETHR$ circuit of size $s$, say it computes $\sgn(c_0 + \sum_{i = 1}^s f_i)$, where $f_i$s have exact threshold representations $\sum_{j = 1}^n w_{i, j}x_j = t_i$, respectively.
Consider the $\THR \circ \THR$ circuit of size $2s$, given by $\sgn\left(\sum_{i = 1}^s c_i (g_{i, 2} - g_{i, 1} + 1)\right)$, where $g_i$s are threshold functions with representations as follows.
\begin{align*}
g_{i, 1} & = 1 \iff \sum_{j = 1}^n w_{i, j}x_j \geq t_i,\\
g_{i, 2} & = 1 \iff \sum_{j = 1}^n w_{i, j}x_j \geq t_i + \epsilon_{f_i}.
\end{align*}
It is easy to verify that this circuit computes $f$.
\end{proof}

\begin{remark}
In fact, Hansen and Podolskii \cite{HP10} showed that the circuit class $\THR \circ \THR$ is identical to the circuit class $\THR \circ \ETHR$.  However, we do not require the full generality of their result.
\end{remark}

We now note that any function computable by a $\THR \circ \OR$ circuit can be computed by a $\THR \circ \AND$ circuit without a blowup in the size.

\begin{lemma}\label{lem: throrand}
Suppose $f : \bra{-1, 1}^n \rightarrow \bra{-1, 1}$ can be computed by a $\THR \circ \OR$ circuit of size $s$.  Then, $f$ can be computed by a $\THR \circ \AND$ circuit of size $s$.
\end{lemma}

\begin{proof}
Consider a $\THR \circ \OR$ circuit of size $s$, computing $f$, say
\[
f(x) = \sgn\left(\sum_{i = 1}^s w_i \bigvee_{j \in S_i}x_j\right)
\]
Note that
\[
\sum_{i = 1}^s w_i \bigvee_{j \in S_i}x_j = \sum_{i = 1}^s -w_i \bigwedge_{j \in S_i}{x_j^c}
\]
Thus, $\sgn\left(\sum_{i = 1}^s -w_i \bigwedge_{j \in S_i}{x_j^c}\right)$ is a $\THR \circ \AND$ representation of $f$, of size $s$.
\end{proof}

\begin{definition}[\textsf{OR} polynomials]
Define a function $p : \bra{-1, 1}^n \rightarrow \R$ of the form $p(x) = \sum_{S \subseteq [n]}a_S \bigvee_{i \in S}x_i$ to be an \emph{\textsf{OR} polynomial}.
Define the weight of $p$ to be $\sum_{S \subseteq [n]}\abs{a_S}$, and its degree to be $\max_{S \subseteq [n]}\bra{\abs{S} : a_S \neq 0}$.
\end{definition}

Define the \emph{sign rank} of a real matrix $A = [A_{ij}]$, denoted by $sr(A)$ to be the least rank of a matrix $B = [B_{ij}]$ such that $A_{ij}B_{ij} > 0$ for all $(i, j)$ such that $A_{ij} \neq 0$.

Forster \cite{Forster01} proved the following relation between the sign rank of a $\bra{\pm 1}$ valued matrix and its spectral norm.
\begin{theorem}[Forster \cite{Forster01}]\label{thm:Forster}
Let $A = [A_{xy}]_{x \in X, y \in Y}$ be a $\bra{\pm 1}$ valued matrix.
Then,
\[
sr(A) \geq \frac{\sqrt{|X||Y|}}{||A||}
\]
\end{theorem}
We require the following generalization of Forster's theorem by Razborov and Sherstov \cite{RS10}.

\begin{theorem}[Razborov and Sherstov \cite{RS10}]\label{thm: RS}
Let $A = [A_{xy}]_{x \in X, y \in Y}$ be a real valued matrix with $s = |X||Y|$ entries, such that $A \neq 0$.  For arbitrary parameters $h, \gamma > 0$, if all but $h$ of the entries of $A$ satisfy $|A_{xy}| \geq \gamma$, then
\[
sr(A) \geq \frac{\gamma s}{||A||\sqrt{s} + \gamma h}
\]
\end{theorem}

The following lemma from Forster et al.~\cite{FKLMSS01} tells us that functions that have efficient $\THR \circ \MAJ$ representations have low sign rank.

\begin{lemma}[Forster et al.~\cite{FKLMSS01}]\label{lem: thrmajsr}
Let $f : \bra{-1, 1}^n \times \bra{-1, 1}^n \rightarrow \bra{-1, 1}$ be a boolean function computed by a $\THR \circ \MAJ$ circuit of size $s$.  Then,
\[
sr(M_f) \leq sn
\]
where $M_f$ denotes the communication matrix of $f$.
\end{lemma}

For the purpose of this paper, we abuse notation, and use $sr(f)$ and $sr(M_f)$ interchangeably, to denote the sign rank of $M_f$.

In the model of communication we consider, two players, say Alice and Bob, are given inputs $X \in \mathcal{X}$ and $Y \in \mathcal{Y}$ for some finite input sets $\mathcal{X}, \mathcal{Y}$, they are given access to \emph{private} randomness and they wish to compute a given function
$f: \mathcal{X} \times \mathcal{Y} \rightarrow \bra{-1, 1}$.  We will use $\mathcal{X} = \mathcal{Y} = \bra{-1, 1}^n$ for the purposes of this paper.
Alice and Bob communicate using a randomized protocol which has been agreed upon in advance.  The cost of the protocol is the maximum number of bits communicated on the worst case input and randomness.
A protocol $\Pi$ computes $f$ with advantage $\epsilon$ if the probability of $f$ agreeing with $\Pi$ is at least $1/2 + \epsilon$ for all inputs.
We denote the cost of the best such protocol to be $R_{\epsilon}(f)$.  Note here that we deviate from the notation used in \cite{KN97}, for example.
Unbounded error communication complexity was introduced by Paturi and Simon \cite{PS86}, and is defined as follows.
\[
\UPP(f) = \min_\epsilon(R_\epsilon(f)).
\]
This measure gives rise to the following natural communication complexity class, as argued by Babai et al.~\cite{BFS86}.
\begin{definition}\label{defn: UPP}
\[
\UPP^{cc}(f) \equiv \bra{f: \UPP(f) = \textnormal{polylog}(n)}.
\]
\end{definition}

Paturi and Simon \cite{PS86} showed an equivalence between $\UPP(f)$ and the sign rank of $M_f$, where $M_f$ denotes the communication matrix of $f$.
\begin{theorem}[Paturi and Simon \cite{PS86}]\label{thm: PS}
For any function $f : \bra{-1, 1}^n \times \bra{-1, 1}^n \rightarrow \bra{-1, 1}$,
\[
\UPP(f) = \log sr(M_f) \pm O(1).
\]
\end{theorem}

The following lemma characterizes the spectral norm of the communication matrix of $\xor$ functions.
\begin{lemma}[Folklore]\label{lem: xor_eigenvalues}
Let $f: \bra{-1, 1}^n \times \bra{-1, 1}^n \rightarrow \mathbb{R}$ be any real valued function and let $M$ denote the communication matrix of $f \circ \xor$.  Then,
\[
||M|| = 2^n \cdot \max_{S \subseteq [n]}\abs{\wh{f}(S)}.
\]
\end{lemma}

Finally, we require the following well-known lemma by Minsky and Papert \cite{MP69}.
\begin{lemma}[Minsky and Papert \cite{MP69}]\label{lem: MP}
Let $p : \bra{-1, 1}^n \rightarrow \R$ be any symmetric real polynomial of degree $d$.
Then, there exists a univariate polynomial $q$ of degree at most $d$, such that for all $x \in \bra{-1, 1}^n$,
\[
p(x) = q(\#1(x))
\]
where $\#1(x) = \abs{\bra{i \in [n] : x_i = 1}}$.
\end{lemma}

\section{Hardness of approximating $\OMB^{0}_l \circ \OR_{m}$}

For notational convenience, denote $g = \OMB^{0}_l, f = g \circ \bigvee_{m}$.  Let $n = lm$.
We first use an idea from Krause and Pudl{\'{a}}k \cite{KP98} which enables us work with $g$, rather than $g \circ \bigvee_m$.

\begin{lemma}\label{lem: thror}
Let $f = g_l \circ \bigvee_{m} : \bra{-1, 1}^{ml} \rightarrow \bra{-1, 1}, \Delta \in \R, e_x \geq 0 ~ \forall x \in X$, where $X$ denotes the set of all inputs $x$ in $\bra{-1, 1}^{ml}$ such that $\bigvee_{m}(x) = -1^l$, and let $p$ be a real polynomial such that
\begin{align*}
\forall x \in \bra{-1, 1}^{ml}, \qquad f(x)p(x) & \geq \Delta,\\
\forall x \in \bra{-1, 1}^{ml} \text{~such that~} \bigvee_{m}(x) = -1^l, \qquad f(x)p(x) & \geq \Delta + e_x.
\end{align*}
Then, there exists an $\OR$ polynomial $p'$, of weight at most $wt(p')$, such that
\begin{align*}
\forall y \in \bra{-1, 1}^l, \qquad p'(y)g(y) & \geq wt(p)\left(\Delta - 2l \cdot 2^{-m}\right)\\
g(-1^l)p'(-1^l) & \geq wt(p)\left(\Delta - 2l \cdot 2^{-m} + \frac{\sum_{x \in X}e_x}{\abs{X}}\right).
\end{align*}
\end{lemma}

\begin{proof}
For any $y \in \bra{-1, 1}^l$, denote by $\E_{y}[f(x)]$ the expected value of $f(x)$ with respect to the uniform distribution over all $x \in \bra{-1, 1}^{ml}$ such that $\bigvee_m(x) = y$.
For each $I_k \subseteq [l] \times [m]$, define $J_k \subseteq [l]$ to be the projection of $I_k$ on $[l]$.  Formally,
\[
i \in J_k \iff \exists j, ~ x_{i, j} \in I_k.
\]
Note that for any $y \in \bra{-1, 1}^l$,
\[
\E_y[f(x)p(x)] = g(y) \cdot \E_y[p(x)] \geq \Delta
\]
and
\[
\E_{-1^l}[f(x)p(x)] = g(-1^l) \cdot \E_{-1^l}[p(x)] \geq \Delta + \frac{\sum_{x \in X}e_x}{\abs{X}}.
\]
Next, we approximate each monomial by an \textsf{OR} function.
The following argument appears in the proof of Lemma 2.3 in \cite{KP98}.  However, we reproduce the proof below for clarity and completeness.

First observe that for all $y \in \bra{-1, 1}^l$, and for all $x$ satisfying $\bigvee_{m}(x) = y$, the monomial corresponding to $I_k$ equals
\[
\bigoplus_{(i, j) \in I_k}x_{i, j} = \bigoplus_{(i, j) \in I_k, y_i = -1} x_{i, j}.
\]
Let $A = \bra{j \in [l]: y_j = -1}$.  If $A \cap J_k = \emptyset$, then
\[
\E_y\left[\bigoplus_{(i, j) \in I_k}x_{i, j}\right] = \bigvee_{j \in J_k}y_j = 1
\]
Else, let $B = A \cap J_k$.  In this case, $\bigvee_{j \in J_k}y_j = -1$.
Also,
\begin{equation}\label{eqn: kpexp}
\E_y\left[\bigoplus_{(i, j) \in I_k}x_{i, j}\right] = \E_{x \in \bra{-1, 1}^{A \cap J_k} : \bigvee(x) = -1^{\abs{A \cap J_k}}}\left[\bigoplus_{(i, j) \in I_k, y_i = -1}x_{i, j}\right]
\end{equation}
Note that
\begin{equation}\label{eqn: a}
\E_{x \in \bra{-1, 1}^{A \cap J_k}}\left[\bigoplus_{(i, j) \in I_k, y_i = -1}x_{i, j}\right] = 0
\end{equation}
Denote $\abs{A \cap J_k} = q$.
Using Equation \ref{eqn: a} and a simple counting argument, the absolute value of the RHS (and thus the LHS) of Equation \ref{eqn: kpexp} can be upper bounded as follows.
\begin{align*}
\abs{\E_y\left[\bigoplus\limits_{(i, j) \in I_k}x_{i, j}\right]} & \leq \frac{2^{mq} - \left(2^{m}-1\right)^q}{\left(2^m - 1\right)^q}\\
& \leq \frac{q \cdot 2^{mq - m}}{2^{mq}} \leq l2^{-m}
\end{align*}
Hence, for all $y \in \bra{-1, 1}^l$, we have
\begin{equation}
\abs{\E_y\left[\bigoplus_{(i, j) \in I_k}x_{i, j}\right] - \frac{1}{2} - \frac{1}{2}\bigvee_{j \in J_k}y_j} \leq 2l2^{-m}.
\end{equation}

Say $p = v_0 + \sum_{k}v_k p_k$, where $p_k(x) = \oplus_{(i, j) \in I_k}x_{i, j}$ is the unique multilinear expansion of $p$.
Define 
\[
p' = v_0 - \frac{\sum_k v_k}{2} - \sum_k\frac{v_k}{2}\bigvee_{j \in J_k}y_j.
\]
Note that 
\[
wt(p') = wt\left(v_0 - \frac{\sum_k v_k}{2} - \sum_{k}\frac{v_k}{2}\bigvee_{j \in J_k}y_j \right) = \abs{v_0 - \frac{\sum_k{v_k}}{2}} + \sum_k\abs{\frac{v_k}{2}} \leq wt(p).
\]

Thus, using linearity of expectation,
we obtain that for all $y \in \bra{-1, 1}^l$,
\[
g(y) \cdot p'(y) \geq \Delta - wt(p)\left(2l \cdot 2^{-m}\right)
\]
and
\[
g(-1^l) \cdot p'(-1^l) \geq \Delta + \frac{\sum_{x \in X}e_x}{\abs{X}} - wt(p)\left(2l \cdot 2^{-m}\right)
\]
\end{proof}

Next, we use random restrictions which reduces the degree of the approximating \textsf{OR} polynomial, at the cost of a small error.

\begin{lemma}\label{lem: lowdeg}
Let $g_l = \OMB^{0} : \bra{-1, 1}^l \rightarrow \bra{-1, 1}, f = g_l \circ \bigvee_m$, and $\Delta, \bra{e_x \geq 0 : x \in X}$ (where $X$ is defined as in Lemma \ref{lem: thror}), and $p$ be a real polynomial such that 
\[
\begin{cases}
\forall x \in \bra{-1, 1}^{ml}, & f(x)p(x) \geq \Delta\\
\forall x \in \bra{-1, 1}^{ml} \text{~such that~} \bigvee_{m}(x) = -1^l, & p(x) \geq \Delta + e_x.\\
\end{cases}
\]
Then, for any integer $d > 0$, there exists an $\OR$ polynomial $p''$, of degree $d$ and weight at most $wt(p)$, such that 
\begin{align*}
\text{For all~} y \in  \bra{-1, 1}^{l/8}, \qquad p''(y)g_{l/8}(y) & \geq \Delta - wt(p)\left(2l \cdot 2^{-m} + 2^{-(d - 1)}\right)\\
\text{and~}\qquad\qquad\qquad p''(-1^{l/8}) & \geq \Delta + \frac{\sum_{x \in X}e_x}{\abs{X}} - wt(p)\left(2l \cdot 2^{-m} + 2^{-(d - 1)}\right).
\end{align*}
\end{lemma}
\begin{proof}

Lemma \ref{lem: thror} guarantees the existence of an $\OR$ polynomial $p'$, of weight at most $wt(p)$, such that 
\begin{align*}
\forall y \in \bra{-1, 1}^l, \qquad p'(y)g(y) & \geq \Delta - wt(p)\left(2l \cdot 2^{-m}\right)\\
p'(-1^l) & \geq \Delta + \frac{\sum_{x \in X}e_x}{\abs{X}} - wt(p)\left(2l \cdot 2^{-m}\right).
\end{align*}

Now, set each of the $l$ variables to $-1$ with probability $1/2$, and leave it unset with probability $1/2$.  Call this random restriction $r$.
Any \textsf{OR} monomial of degree at least $d$ gets fixed to $-1$ with probability $1 - 2^{-d}$.  Thus, by linearity of expectation, the expected weight of surviving monomials of degree at least $d$ in $p'$ is at most $wt(p)\cdot 2^{-d}$.  Let $M|_r$ denote the value of a monomial $M$ after the restriction $r$.
By Markov's inequality,
\[
\Pr_r\left[\sum_{M : \deg(M|_r) \geq d}wt(M|_r) > wt(p)\cdot 2^{-d + 1}\right] < 1/2
\]
Consider $l/2$ pairs of variables, $\bra{(x_i, x_{i + 1}) : i \in [l/2]}$ (assume w.l.o.g that $l$ is even).
For any pair, the probability that both of its variables remain unset is $1/4$.
This probability is independent over pairs.
Thus, by a Chernoff bound, the probability that at most $l/16$ pairs remain unset is at most $2^{-\frac{l}{64}}$.

By a union bound, there exists a setting of variables such that at least $l/16$ pairs of variables are left free, and the weight of degree $\geq d$ monomials in $p'$ is at most $wt(p) \cdot 2^{-d + 1}$.  Set the remaining $7l/8$ variables to the value $-1$.  After the restriction, drop the monomials of degree $\geq d$ from $p'$ to obtain $p''$, which is now an \textsf{OR} polynomial of degree less than $d$ and weight at most $wt(p)$.  Note that the function $g_l$ hit with this restriction just becomes $g_{l/8}$.

Thus,
\begin{align*}
\text{For all~} y \in  \bra{-1, 1}^{l/8}, \qquad p''(y)g_{l/8}(y) & \geq \Delta - wt(p)\left(2l \cdot 2^{-m} + 2^{-(d - 1)}\right)\\
\text{and~}\qquad\qquad\qquad p''(-1^{l/8}) & \geq \Delta + \frac{\sum_{x \in X}e_x}{\abs{X}} - wt(p)\left(2l \cdot 2^{-m} + 2^{-(d - 1)}\right).
\end{align*}

\end{proof}

\subsection{Hardness of $\OMB^0$}
In this section, we show that approximating $\OMB^0$ by a low weight polynomial $p$ must imply that the degree of $p$ is large.

We require the following result by Ehlich and Zeller \cite{EZ64} and Rivlin and Cheney \cite{RC66}.

\begin{lemma}[\cite{EZ64, RC66}]\label{lem: approx}
The following holds true for any real valued $\alpha > 0$ and $k > 0$.  Let $p$ be a univariate polynomial of degree $d < \sqrt{k/4}$, such that $p(0) \geq \alpha$, and $p(i) \leq 0$ for all $i \in [k]$.  Then, there exists $i \in [k]$ such that $p(i) < -2\alpha$.
\end{lemma}

We next use the idea of `doubling' for the $\OMB^0$ function, as in \cite{Beigel94, BVW07} to show that a low degree polynomial of bounded weight cannot represent $\OMB^0$ well.  This is our main approximation theoretic lemma.

\begin{lemma}\label{lem: doubling}
Suppose $p$ is a polynomial of degree $d < \sqrt{n/4}$ and $a > 0, b \in \R$ be reals such that $\OMB^0(-1^n) \geq a$ and $\OMB^0(x)p(x) \geq b$ for all $x \in \bra{-1, 1}^n$.  Then, for all $i \in \bra{0, 1, \dots, \lfloor n/10d^2\rfloor}$, there exists an $x_i \in \bra{-1, 1}^n$ (not necessarily distinct) such that $\abs{p(x_i)} \geq 2^{i}a + \left(3 \cdot 2^{i} - 3\right)b$.
\end{lemma}
The argument will be an iterative one, inspired by the arguments of Beigel and Buhrman et al.~\cite{Beigel94, BVW07}.

\begin{claim}\label{claim: dummy}
If $a$ and $b$ are reals such that $a > 0, b \in \R$ and $2^i a + \left(3 \cdot 2^{i} - 2\right)b < 0$ for some $i \geq 0$, then 
$2^j a + \left(3\cdot 2^{j} - 3\right)b < 0$ for all $j > i$.
\end{claim}
\begin{proof}
Note that since $a > 0$ and $2^i a + \left(3\cdot 2^{i} - 2 \right)b < 0$, $b$ must be negative.
For any $j > i$, write $2^j a + \left(3 \cdot 2^{j} - 3\right)b = 2^{j - i}\left( 2^i a + \left(3 \cdot 2^{i} - 2\right)b\right) + 3\cdot(2^{j - i + 1} - 3)b < 0$.
\end{proof}

\begin{proof}[Proof of Lemma \ref{lem: doubling}]

We will assume, for the rest of the proof, that 
\begin{equation}\label{eqn: positive}
2^i a + \left(3 \cdot 2^{i} - 2\right)b \geq 0 ~\forall i \in \left[\lfloor n/10d^2 \rfloor\right].
\end{equation}
If not, the lemma is trivially true by Claim \ref{claim: dummy}.

Divide the variables into $\lfloor n/10d^2 \rfloor$ contiguous blocks of size $10d^2$ each.

\textbf{Induction hypothesis:}  For each $i \in \bra{1, \dots, \lfloor n/10d^2 \rfloor}$, there exists an input $x^i \in \bra{-1, 1}^n$ such that
\begin{itemize}
\item $x^i_j = -1$ for all indices to the right of the $i$th block.
\item The values of $x^i_j$ for indices $j$ to the left of the $i$th block are set as dictated by the previous step.
\item $\abs{p(x)} \geq 2^ia + \left(3 \cdot 2^{i} - 3\right)b$.
\item The value of $p(x)$ is negative if $i$ is odd, and positive if $i$ is even.
\end{itemize}
We now prove the induction hypothesis.

\begin{itemize}
\item \textbf{Base case:}  Say $i = 1$.
We know from our assumption that $\OMB^0(-1^n) \geq a$ and $\OMB^0(x)p(x) \geq b$ for all $x \in \bra{-1, 1}^n$.
Set the variables corresponding to the even indices in the first block to $-1$, and all variables to the right of the first block to $-1$.
Denote the free variables by $y_1, \dots, y_{5d^2}$.
Define a polynomial $p_1 : \bra{-1, 1}^{5d^2} \rightarrow \R$ by $p_1(y) = \E_{\sigma \in S_{5d^2}}\tilde{p}(\sigma(y))$, where $\tilde{p}(y)$ denotes the value of $p$ on input $y_1, \dots, y_{5d^2}$, and the remaining variables are set as described earlier. The expectation is over the uniform distribution.
Note that $p_1$ is a symmetric polynomial of degree at most $d$, and satisfies
\[
p_1(-1^{5d^2}) \geq a, \qquad p_1(y) \leq -b ~ \forall y \neq -1^{5d^2}.
\]

By Lemma \ref{lem: MP}, there exists a univariate polynomial $p_1'$ such that for all $i \in \bra{0} \cup [5d^2]$,
\[
p_1'(i) = p_1(y) ~ \forall y ~ \text{such that} ~ \#1(y) = i 
\]
Thus,
\[
p_1'(0) \geq a, \qquad p_1'(j) \leq -b ~ \forall j \in [5d^2].
\]
Define $p_1'' = p_1' + b$.
Thus, $p_1''(0) \geq a + b \geq 0$, and $p_1''(j) \leq 0 ~ \forall j \in [5d^2]$.

By Lemma \ref{lem: approx}, there exists a $j \in [5d^2]$ such that $p_1''(j) < -2a - 2b$.  This means $p_1'(j) < -2a - 3b < 0$, because of Equation \ref{eqn: positive}.  This implies existence of an $x \in \bra{-1, 1}^n$ (with all variables to the right of the first block still set to $-1$) such that $p(x) < -2a - 3b$.

\item \textbf{Inductive step:}
In the $i$th block, set the variables corresponding to the even indices to $-1$ if $i$ is odd, and set the odd indexed variables to $-1$ if $i$ is even.  Set the variables outside the $i$th block as dictated by the previous step.
Assume that $i$ is odd (the argument for even integers $i$ follows in a similar fashion, with suitable sign changes).
Denote the free variables by $y_1, \dots, y_{5d^2}$.
Define a polynomial $p_i : \bra{-1, 1}^{5d^2} \rightarrow \R$ by $p_i(y) = \E_{\sigma \in S_{5d^2}}\tilde{p}(\sigma(y))$, where $\tilde{p}(y)$ denotes the value of $p$ on input $y_1, \dots, y_{5d^2}$, and the remaining variables are set as described earlier. The expectation is over the uniform distribution.
Note that $p_i$ is a symmetric polynomial of degree at most $d$, and satisfies
\[
p_i(-1^{5d^2}) \geq 2^ia + \left(3 \cdot 2^i - 3\right)b, \qquad p_1(y) \leq -b ~ \forall y \neq -1^{5d^2}.
\]
By Lemma \ref{lem: MP}, there exists a univariate polynomial $p_i'$ such that for all $j \in \bra{0} \cup [5d^2]$,
\[
p_i'(j) = p_i(y) ~ \forall y ~ \text{such that} ~ \#1(y) = j.
\]
Thus,
\[
p_i'(0) \geq 2^ia + \left(3 \cdot 2^i - 3\right)b, \qquad p_1'(j) \leq -b ~ \forall j \in [5d^2].
\]
Define $p_i'' = p_i' + b$.
Thus, 
\[
p_i''(0) \geq 2^ia + \left(3 \cdot 2^i - 2\right)b \geq 0, \qquad p_i''(j) \leq 0 ~ \forall j \in [5d^2].
\]

By Lemma \ref{lem: approx}, there exists a $j \in [5d^2]$ such that $p_i''(j) \leq -2^{i + 1}a - \left(3 \cdot 2^{i + 1} - 2\right)b$, and hence $p_i'(j) \leq -2^{i + 1}a - \left(3 \cdot 2^{i + 1} - 3\right)b$, by Equation \ref{eqn: positive}.  This implies the existence of an $x$ in $\bra{-1, 1}^n$ (with all variables to the right of the $i$th block still set to $-1$, and variables to the left of the $i$th block as dictated by the previous step) such that $p(x) < -2^{i + 1}a - \left(3 \cdot 2^{i + 1} - 3\right)b$.
\end{itemize}
\end{proof}
\section{Lower bounds}
In this section, we prove our lower bounds.
We first use linear programming duality to give us a sufficient approximation theoretic condition $f$ for showing that the sign rank of $f \circ \xor$ is large.
Let $\delta > 0$ be a parameter, and $X$ be any subset of $\bra{-1, 1}^n$.

\begin{center}\hypertarget{LP1}{}
(LP1) \framebox{
\begin{tabular}{llllll}
Variables & $\epsilon, \bra{\mu_x: x \in \bra{-1, 1}^n}$ &               &  &                             &  \\
Minimize  & $\epsilon$                                &               &  &                             &  \\
s.t.      & $\abs{\sum\limits_{x}\mu(x)f(x)\chi_S(x)}$ & $\leq \epsilon$ &  & $\forall S \subseteq [n]$ &  \\
          & $\sum\limits_{x}\mu(x)$             &  $= 1$             &  &                             &  \\
          & $\epsilon \geq 0$                                   &              &  &                             &  \\
          & $\mu(x) \geq \frac{\delta}{2^n}$                                        &               &  &$\forall x \in X$                             & 
\end{tabular}}
\end{center}

The first two constraints above specify that correlation of $f$ against \emph{all parities} need to be small w.r.t a distribution $\mu$. The last constraint enforces the fact that $\mu$ is `$\delta$-smooth' over the set $X$. As we had indicated before in Section~\ref{sec:techniques}, these constraints make analyzing the LP challenging.

Standard manipulations (as in \cite{CM17}, for example) and strong linear programming duality reveal that the optimum of the above linear program equals the optimum of the following program.

\begin{center}\hypertarget{LP2}{}
(LP2) \framebox{
\begin{tabular}{llllll}

Variables & $\Delta, \{\alpha_S: S \subseteq [n]\}, \bra{\xi_x : x \in X}$ &               &  &                             &  \\
Maximize  & $\Delta + \frac{\delta}{2^n}\sum\limits_{x \in X}\xi_x$                                &               &  &                             &  \\
s.t.      & $f(x) \sum\limits_{S \subseteq [n]}\alpha_S\chi_S(x)$ & $\geq \Delta$ &  & $\forall x \in \{-1, 1\}^n$ &  \\
          & $f(x) \sum\limits_{S \subseteq [n]}\alpha_S\chi_S(x)$ & $\geq \Delta + \xi_x$ &  & $\forall x \in X$ &  \\
          & $\sum\limits_{S \subseteq [n]}|\alpha_S|$             &  $\leq 1$             &  &                             &  \\
          & $\Delta \in \mathbb{R}$                                   &              &  &                             &  \\
          & $\alpha_S \in \mathbb{R}$                                        &               &  &$\forall S \subseteq [n]$                             & \\
          & $\xi_x \geq 0$ & & & $\forall x \in X$\\
\end{tabular}}
\end{center}

The variable $\Delta$ represents the worst margin guaranteed to exist at all points. At each point $x$ over the smooth set, the dual polynomial has to better the worst margin by at least $\xi_x$. If the OPT is large, then it means that on average the dual polynomial did significantly better than the worst margin. Below is our main technical result of this section, which says that no such dual polynomial exists, even when the smoothness parameter $\delta$ is as high as 1/4.

\begin{theorem}\label{thm: upplb}
Let $f = \OMB^{0}_l \circ \bigvee_{l^{1/3} - \log l} : \bra{-1, 1}^{l^{4/3} - l \log l} \rightarrow \bra{-1, 1}, \delta = 1/4$
and $X = \bra{x \in \bra{-1, 1}^{l^{4/3} - l \log l} : \bigvee(x) = -1^l}$.
Then the optimal value, $\OPT$, of \hyperlink{LP2}{(LP2)} is at most $2^{-\frac{l^{1/3}}{81}}$.
\end{theorem}
\begin{proof}

Let $p$ be a polynomial of weight $1$, for which \hyperlink{LP2}{(LP2)} attains its optimum. Denote the values taken by the variables at the optimum by $\Delta_{\OPT}, \bra{\xi_{x, \OPT} : x \in X}$.
Towards a contradiction, assume $\mathsf{OPT} \geq 2^{-\frac{l^{1/3}}{81}}$.

Lemma \ref{lem: lowdeg} (set $m = l^{1/3} - \log l$) shows the existence of an \textsf{OR} polynomial $p'$ on $l/8$ variables, of weight $1$, such that 
\begin{align*}
\text{For all~} y \in  \bra{-1, 1}^{l/8}, \qquad p'(y)\OMB^0(y) & \geq \Delta_\OPT - 2 \cdot 2^{-l^{1/3}} - 2 \cdot 2^{-l^{1/3}}\\
\text{and~}\qquad\qquad\qquad p'(-1^{l/8}) & \geq \Delta + \frac{\sum_{x \in X}\xi_{x, \OPT}}{\abs{X}} - 2 \cdot 2^{-l^{1/3}} - 2 \cdot 2^{-l^{1/3}}.
\end{align*}
Note that 
\begin{equation}\label{eqn: deltalb}
\mathsf{OPT} \geq 2^{-\frac{l^{1/3}}{81}} \implies \Delta_\OPT \geq 2^{-\frac{l^{1/3}}{81}} - \delta\frac{\sum_{x \in X}\xi_{x, \OPT}}{2^n}
\end{equation}
$p'$ satisfies the assumptions of Lemma \ref{lem: doubling} with $d = \deg(p') = l^{1/3} < \sqrt{l/32}$ (since any $\OR$ polynomial of degree $d$ can be represented by a polynomial of degree at most $d$), $a = \Delta_\OPT + \frac{\sum_{x \in X}\xi_{x, \OPT}}{|X|} - 4 \cdot 2^{-l^{1/3}}$, and $b = \Delta_\OPT - 4 \cdot 2^{-l^{1/3}}$.
\begin{align*}
a & = \Delta_\OPT + \frac{\sum_{x \in X}\xi_{x, \OPT}}{|X|} - 4 \cdot 2^{-l^{1/3}}\\
& \geq 2^{-\frac{l^{1/3}}{81}} - 4 \cdot 2^{-l^{1/3}} \geq 0.
\end{align*}

Let us denote $k = l^{1/3}/80$ for the remaining of this proof.
Thus, by Lemma \ref{lem: doubling}, there exists an $x \in \bra{-1, 1}^{l/8}$ such that
\begin{align*}
\abs{p'(x)} & \geq 2^ka +  \left(3 \cdot 2^k - 3\right)b\\
& \geq \Delta_\OPT(4 \cdot 2^k - 3) + 2^k\frac{\sum_{x \in X}\xi_{x, \OPT}}{|X|} - 4 \cdot 2^{-80k}(4 \cdot 2^k - 3)\\
& \geq \left(4 \cdot 2^k - 3\right)\left(2^{-\frac{l^{1/3}}{81}} - \delta\frac{\sum_{x \in X}\xi_{x, \OPT}}{2^n}\right) + 2^k\frac{\sum_{x \in X}\xi_{x, \OPT}}{|X|} - 4 \cdot 2^{-80k}(4 \cdot 2^k - 3)\tag*{Using Equation \ref{eqn: deltalb}.}\\
& \geq \left(4 \cdot 2^k - 3\right)\left(2^{-80k/81} - 4 \cdot 2^{-80k}\right) \tag*{Since $\delta = 1/4$.}\\
& > 1 \tag*{Assuming $k > 81$.}
\end{align*}
This yields a contradiction, since $p'$ was a polynomial of weight at most $1$.
\end{proof}

\begin{theorem}\label{thm: sr}
Let $f = \OMB^{0}_l \circ \bigvee_{l^{1/3} - \log l} : \bra{-1, 1}^{l^{4/3} - l \log l} \rightarrow \bra{-1, 1}$.
Then,
\[
sr(f \circ \xor) \geq \frac{2^{l^{1/3} - 2 \log l}}{16}
\]
\end{theorem}

\begin{proof}
Let $n = l^{4/3} - l \log l$.
Theorem \ref{thm: upplb} tells us that the optimum of \hyperlink{LP2}{(LP2)} (and hence \hyperlink{LP1}{(LP1)}, by duality) is at most $2^{-\frac{l^{1/3}}{81}}$, when $f = \OMB^{0} \circ \bigvee_{l^{1/3} - \log l}$.
We first estimate the size of $X^c$.  The probability (over the uniform distribution on the inputs) of a particular $\OR$ gate firing a 1 is $\frac{1}{2^{l^{1/3} - \log l}}$.  By a union bound, the probability of any $\OR$ gate firing a 1 is at most $\frac{l^2}{2^{l^{1/3}}}$, and hence $|X^c| \leq 2^n \cdot \frac{l^2}{2^{l^{1/3}}}$.
By Lemma \ref{lem: xor_eigenvalues} and Theorem \ref{thm: RS},
\begin{align*}
sr(f \circ \xor) & \geq sr(f\mu \circ \xor) \geq \frac{\frac{\delta}{2^n} 2^{2n}}{\OPT \cdot 2^n + \frac{\delta}{2^n} \cdot h}\\
& \geq \frac{1/4}{2^{-\frac{l^{1/3}}{81}} + \frac{1}{4}\frac{\abs{X^c}}{2^n}}\\
& \geq \frac{1/4}{2^{-\frac{l^{1/3}}{81}} + \frac{1}{4}l^2 \cdot 2^{-l^{1/3}}}\\
& \geq \frac{2^{l^{1/3} - 2 \log l}}{8}
\end{align*}
This gives us a function $f$ on $n$ input variables such that for large enough $n$,
\[
sr(f \circ \xor) \geq \frac{2^{n^{1/4} - \frac{3}{2}{\log n}}}{8}
\]
\end{proof}
\begin{corollary}
Let $f = \OMB^{0}_l \circ \bigvee_{l^{1/3} - \log l} : \bra{-1, 1}^{l^{4/3} - l \log l} \rightarrow \bra{-1, 1}$, and let $n = l^{4/3} - l \log l$ denote the number of input variables.
Then
\[
\UPP(f \circ \xor) \geq n^{1/4} - \frac{3}{2}\log n - 3.
\]
\end{corollary}
\begin{proof}
It follows from Theorem \ref{thm: sr} and Theorem \ref{thm: PS}.
\end{proof}
We now prove Theorem \ref{thm:omb-ckt-size}, which gives us a lower bound on the size of $\THR \circ \MAJ$ circuits computing $\OMB^{0} \circ \bigvee_{l^{1/3} - \log l} \circ \XOR_2$.
\begin{proof}[Proof of Theorem \ref{thm:omb-ckt-size}.]
Suppose $\OMB^{0} \circ \bigvee_{l^{1/3} - \log l} \circ \XOR_2$ could be represented by a $\THR \circ \MAJ$ circuit of size $s$.
Let $n = 2l^{4/3} - 2l \log l$.  By Lemma \ref{lem: thrmajsr} and Theorem \ref{thm: sr},
\[
s\left(2l^{4/3} - 2l \log l\right) \geq sr(f) \geq \frac{2^{l^{1/3} - 2 \log l}}{8}.
\]
Thus,
\[
s \geq \frac{2^{l^{1/3} - \frac{10}{3} \log l}}{16} = 2^{\Omega\left(n^{1/4}\right)}.
\]
\end{proof}

Finally, we prove Corollary \ref{cor:main}, which separates $\THR \circ \MAJ$ from $\THR \circ \THR$.
\begin{proof}[Proof of Corollary \ref{cor:main}.]
Let $n = 2l^{4/3} - 2l \log l$.  By Lemma \ref{lem: throrand}, $f = \OMB^{0} \circ \bigvee_{l^{1/3} - \log l} \circ \XOR_2$ can be computed by a $\THR \circ \AND \circ \XOR_2$ circuit of size $n$. Hence $f \in \THR \circ \ETHR = \THR \circ \THR$, by Theorem \ref{thm: HP}.
By Theorem \ref{thm:omb-ckt-size}, $\THR \circ \MAJ$ circuits computing $f$ require size $2^{\Omega\left(n^{1/4}\right)}$.
\end{proof}

\section{Conclusions}
This work refines our understanding of depth-2 threshold circuits by providing the following summary: 
$$\widehat{LT}_1 \subsetneq LT_1 \subsetneq \widehat{LT}_2 = \MAJ \circ \THR \subsetneq \THR \circ \MAJ \subsetneq LT_2 \subseteq \widehat{LT}_3 \subseteq \mathrm{NP}/\mathrm{poly}$$
While we cannot rule out that SAT has efficient $\THR \circ \THR$ circuits, we do not even know whether \textsf{IP} is in $LT_2$. On the other hand, the most powerful method used to prove lower bounds on the size of depth-2 threshold circuits for computing an explicit function $f$ exploits the fact that $f$ has large sign rank. Before our work, it was not known if $LT_2$ contained any function of large sign rank. Our main result shows that indeed there are such functions, answering a question explicitly raised by Hansen and Podolskii \cite{HP10} and Amano and Maruoka \cite{AM05}. 

The central open question in the area is to prove super-polynomial lower bounds on the size of $\THR \circ \THR$ circuits. The best known explicit lower bounds due to Kane and Williams \cite{KW16} is roughly $n^{3/2}$. We feel that there is a dire need of discovering new techniques for proving strong lower bounds against $\THR \circ \THR$ circuits.

\section*{Acknowledgements}
We are grateful to Kristoffer Hansen for bringing to our attention the question of separating the classes $\THR \circ \MAJ$ and $\THR \circ \THR$ at the summer school on lower bounds, held in Prague in 2015. We also thank Michal Kouck{\'{y}} for organizing and inviting us to the workshop.

\bibliography{bibo}

\end{document}